\documentclass[conference,a4paper]{IEEEtran}
\usepackage{latexsym, amssymb, amsfonts,amsmath, amstext,graphicx, color, bm, times, relsize, mathptm, soul,amsthm,float}

\def\Tr{\mathop{\rm Tr}\nolimits}

\newcommand{\map}[1]{\mathcal{#1}}

\newcommand{\set}[1]{\mathsf{#1}}
\newcommand{\grp}[1]{\mathsf{#1}}

\def\>{\rangle}
\def\<{\langle}

{\catcode`\|=\active
  \gdef\Braket#1{\begingroup
\mathcode`\|32768\let|\BraVert\left<{#1}\right>\endgroup}}
\def\BraVert{\egroup\,\mid\,\bgroup}

\newtheorem{Theorem}{Theorem}
\newtheorem{Lemma}{Lemma}

\definecolor{kmblue}{rgb}{0.19, 0.25, 0.91}
\definecolor{kmred}{rgb}{0.79, 0.29, 0.0}
\definecolor{kmgreen}{rgb}{0, 0.42, 0.24}

\begin{document}

\title{Compression for qubit clocks}

\author{
 \IEEEauthorblockN{Yuxiang Yang$^{a}$, Giulio Chiribella$^{b,c}$ and Masahito Hayashi$^{d,e}$}
  \IEEEauthorblockA{
$~^{a}$Department of Computer Science, The University of Hong Kong \\
$~^{b}$Department of Computer Science, The University of Oxford \\
$~^c$ Canadian Institute for Advanced Research, CIFAR Program in Quantum Information Science\\
$~^{d}$Graduate School of Mathematics, Nagoya University \\
$~^{e}$ Centre for Quantum Technologies, National University of Singapore \\
    Email: {yangyx09@hku.hk, giulio.chiribella@cs.ox.ac.uk \& masahito@math.nagoya-u.ac.jp} }
}

\maketitle


\begin{abstract}  
Two-level (qubit) clock systems are often used to perform precise measurement of time. In this work, we propose a compression protocol for $n$ identically prepared states of qubit clocks. The protocol faithfully encodes the states into $(1/2)\log n$ qubits and $(1/2)\log n$ classical bits and works even in the presence of noise. If the purity of the clock states is fixed, $(1/2)\log n$ qubits are sufficient. We also prove that this protocol requires the minimum amount of total memory  among all protocols with vanishing error in the large $n$ limit.
\end{abstract}


\section{Introduction.}

Atomic clocks, which consist of thousands of identical clock qubits, have already been a mature technique for years. In the International System of Units, for instance, the unit of time is defined by the oscillation frequency between two hyperfine energy levels of the Cs$^{133}$ atom \cite{clock-standard}, while recent developments on optical atomic clocks \cite{optical-clock} promise even higher resolution.
A long sequence of protocols based on qubit clock states and techniques in quantum information processing  have been proposed \cite{chuang-clock,jozsa-clock,lukin-clock,network-nature,network-prl}, which have applications in GPS \cite{gps}, frequency standard \cite{fs} and astronomy \cite{astronomy1,astronomy2}. An efficient compression, therefore, will reduce the communication cost in these protocols.

In this work, we design compressors for $n$ identically prepared qubit clock states, which can be mixed due to noisy evolution. Compared to other compressors proposed by the authors \cite{yang-chiribella-2016-prl,universal,stopwatch}, the compressors here are tailor-made for qubit clock states and achieve the information theoretical limit for the total memory cost. The memory cost is $(1/2)\log n$ qubits in the leading order, with the same amount of ancillary classical bits if the purity of the clock state is undetermined. 
The cost matches the general statement in \cite{population}, which says that $(1/2)\log n$ (qu)bits are needed per degree of freedom.
Compared to  \cite{population}, the protocols here  are constructed using a completely different approach from the protocols, tailor-made for qubit clock states: The second order term of its memory cost is $O(\log\log n)$ in contrast to $x\log n$ (for any positive $x$) as in \cite{population}. The error of the protocols here also vanishes faster than that of the protocols in \cite{population} as $n$ grows large.
The optimality of the protocols here also holds in a stronger sense than in \cite{population}. Indeed, we show that any protocol with less memory than $(1/2-\delta)\log n$ (qu)bits of memory per degree of freedom with $\delta>0$ must have maximum error.

\section{Main result.}

Qubit clock states are basic units of time in quantum information. For instance, in quantum estimation theory one frequently considers pure clock states, while in reality the clock state could be mixed in the presence of noise. In general,
\emph{qubit clock states} are states of the form
\begin{align}\label{clock-define}
\rho_{t,p}&:=p\, |\phi_t\>\<\phi_t|+(1-p)\, |\phi_{t,\perp}\>\<\phi_{t,\perp}| \, \qquad t\in[0,2\pi),
\end{align}
with $p\in(1/2,1]$, $|\phi_t\>=\sqrt{s}|0\>+\sqrt{1-s}e^{it}|1\>$, and $|\phi_{t,\perp}\>:=\sqrt{1-s}|0\>-\sqrt{s}e^{it}|1\>$ for some fixed $s\in(0,1)$. We call $p$ the spectrum since it determines the spectrum of the clock state. If the clock state starts in a pure state ($p=1$) and goes through, for example, depolarizing time evolution, it will end up in the form (\ref{clock-define}) with $p=(e^{-\gamma t}+1)/2$, where $\gamma>0$ is the depolarizing parameter.

The task considered here is the compression of $n$ identical copies of a qubit clock state $\rho_{t,p}$, with the time parameter $t$ unknown. Such a task requires us to design a compressor, which consists of two components (both characterized by completely positive trace-preserving linear maps): the encoder $\map{E}$, which compresses the input state into a memory of the smallest possible size, and the decoder $\map{D}$, which recovers the state from the memory. 
To avoid distortion of time information, we require the compressor to be  {\em faithful}, which means that its error vanishes in the large $n$ limit. We choose as a measure of error the worst case trace distance between the original state and the recovered state $\map{D}\circ\map{E}(\rho_{t,p}^{\otimes n})$
\begin{align}
\epsilon:=\sup_{t}\frac12\|\rho_{t,p}^{\otimes n}-\map{D}\circ\map{E}(\rho_{t,p}^{\otimes n})\|_1 \, .
\end{align}
The main result of this work characterizes the minimum amount of memory needed in a faithful compression of quantum information of time:
\begin{Theorem}
$n$ identical copies of a qubit clock state (\ref{clock-define}) can be optimally compressed into $(1/2)\log n+O(\log\log n)$ qubits in a faithful fashion when its spectrum $p$ is known, $(1/2+x)\log n$ ($x$ is an arbitrary positive constant) additional classical bits are required if its spectrum $p$ is unknown.
\end{Theorem}

\section{Qubit clock compressors.}
Here we present compressors with reduced classical memory size, for clock states defined by Eq. (\ref{clock-define}).  We distinguish between two cases: One is the case of known spectrum,  where the eigenvalue  $p$ has a fixed  value $p_0$, known {\em a priori}.   This is the case when the state evolves noiselessly. The other case is the case of unknown spectrum, where prior knowledge  of $p$ is not assumed (except that $p$ is not equal to $1/2$). We also show that both compressors are faithful.

Before precisely defining the compressors, it is convenient to introduce some properties of qubit clock states. First, the $n$-fold product state $\rho_{t,p}^{\otimes n}$ can be decomposed as 
\begin{align}\label{decomp}
\rho_{t,p}^{\otimes n}\simeq\sum_{J=0}^{n/2}q_J\left(|J\>\<J|\otimes\rho_{t,p,J}\otimes\frac{I_{m_J}}{m_J}\right) \, , 
\end{align}
where $\simeq$ denotes the unitary equivalence implemented by the Schur transform \cite{bacon-chuang-2006-prl,harrow2005applications}, $J$ is the quantum number of the total spin, $q_{J}$ is a probability distribution,  $|J\>$ is the state of the index register,  $\rho_{t,p,J}$ is the state of the representation register,  and  $I_{m_J}/m_J$  is the maximally mixed state in a suitable subspace of the multiplicity register    \cite{fulton-harris,book-hayashi}.  
The state $\rho_{t,p,J}$ can be expressed in the form 
\begin{align}
\rho_{t,p,J}=U_t^{\otimes n}\left(\rho_{p,J}\right)U_t^{\dag \, \otimes n}\qquad U_t=|0\>\<0|+e^{it}|1\>\<1|,
\end{align}
where the fixed state  $\rho_{p,J}$ has the form
\begin{align}
&\rho_{p,J}:=(N_{J})^{-1}\sum_{m=-J}^{J}p^{J+m}(1-p)^{J-m}|J,m\>_s\<J,m|_s\\
&N_J:=\sum_{k=-J}^{J}p^{J+m}(1-p)^{J-m}
\end{align}
where $|J,m\>_s$ is the orthonormal basis defined as
\begin{align}\label{sym}
|J,m\>_s:=\frac{\sum_{\pi\in\grp{S}_{2J}}V_\pi |\phi_0\>^{\otimes (J+m)}|\phi_{0,\perp}\>^{\otimes (J-m)}}{\sqrt{(2J)!(J+m)!(J-m)!}}
\end{align}
with $|\phi_0\>=\sqrt{s}|0\>+\sqrt{1-s}|1\>$, $|\phi_{0,\perp}\>=\sqrt{1-s}|0\>-\sqrt{s}|1\>$, $\grp{S}_{2J}$ being the $(2J)$-symmetric group and $V_\pi$ being the unitary implementing the permutation $\pi$. 

As a key ingredient to reduce the quantum cost of storing $\rho_{t,p,J}$, we introduce a class of quantum operations called the frequency projection channels.
The frequency projection channel $\map{P}_{{\rm proj}, J}$ is defined as
\begin{align}
\map{P}_{{\rm proj},J}(\rho)&:=P_{{\rm proj},J}\,\rho\,P_{{\rm proj},J}+\left(1-\Tr[\rho\,P_{{\rm proj},J}]\right)\rho_0\nonumber \\
P_{{\rm proj},J}&:=\sum_{|m-(2s-1)J|\le \frac{\sqrt{J}\log J}2}|J,m\>\<J,m|
\label{define-setC},
\end{align}
where $\rho_0$ is a fixed state of the representation register. It can be seen that the frequency projection channels, when applied to $\rho_{t,p,J}$, cut down their size by almost half. Furthermore, the projection is faithful for large enough $J$: 
\begin{Lemma}[\cite{stopwatch}]\label{lemma-proj}
For large $J$, the frequency projection error $\epsilon_{{\rm proj},J}:=\frac12\left\|\map{P}_{{\rm proj},J}(\rho_{t,p,J})-\rho_{t,p,J}\right\|_1$ is upper bounded as
\begin{align}
\epsilon_{{\rm proj},J}\le (3/2)J^{{-\frac18\ln\left(\frac{p}{1-p}\right)}}+O\left(J^{-\frac18\ln J}\right)
\end{align}
 for every $t$.
\end{Lemma}

\medskip
\noindent{\bf Description of the compressors.}
For known spectrum, the compressor works following the instruction below:
\begin{itemize}
\item {\em Encoder.} 
Define the operation $\map{C}_{J\to K}$ as the concatenation of the following operations: \\
$i)$ for an input state of a spin-$J$ system, encode the state into $2J$ qubits using the isometry $V$ that maps the basis of the spin-$J$ system into the symmetric basis of $2J$ qubits; \\
$ii)$ apply the optimal universal cloner \cite{werner} from $2J$ qubits to $2K$ qubits if $J\ge K$ or discard $2(J-K)$ qubits if $J>K$; \\
$iii)$ encode the state back into the representation register using the inverse of the isometry $V$.\\
First perform the Schur transform and measure the state of the index register. For outcome $J$, apply the operation $\map{C}_{J\to J_0}$  to the state of the representation register with $J_0=(p-1/2)(n+1)$.
Then apply to the output state of the cloner the frequency projection $\map{P}_{{\rm proj}, J_0}$ defined by Eq. (\ref{define-setC}). Encode the state after the projection into a quantum memory.

\item {\em Decoder.} Sample a value $K$ with the probability distribution $q_K$ and apply the operation $\map{C}_{J_0\to K}$ to the state of the quantum memory. Then append the state $I_{m_{K}}/m_{K}$ to the output. Finally, perform the inverse of the Schur transform to get the recovered state. 

\end{itemize}
The compressor only requires a quantum memory of $\log (\sqrt{J_0}\log J_0+1)$ qubits, which is upper bounded by $(1/2)\log n+\log\log n+1$. 

For unknown spectrum, the compressor requires an additional procedure to discretize and encode the spectrum.
For this purpose, we define a partition of the set $\{0,\dots,n/2\}$ into $b=\lfloor n^{1/2+x}\rfloor$ intervals $\set L_1, \dots \set L_b$ for $x>0$, defined as follows:
\begin{align*}
\set{L}_m&=\{  (m-1) \,r ,\dots, m \, r  -1 \} \qquad \ m=  1,\dots,  b-2 \\
\set{L}_{b-1}&=\{0,\dots,n/2-1\}\setminus\cup_{i=1}^{b-2}\set{L}_i \\
\set{L}_{b}&=\{n/2\}
\end{align*} 
 where $r$ is an integer defined by the relation $r(b-2)< n/2-1\le r(b-1)$ so that $$r=O\left(n^{1/2-x}\right)$$ is either equal to ($m\le b-2$) or larger than ($m=b-1,b$) the width of each interval $\set{L}_m$. Note that  the element $n/2$ is singled out to ensure that the protocol works well for pure states ($p=1$).
We also denote by $$\set{Med}=\left\{\left\lfloor \frac r2\right\rfloor, \left\lfloor \frac{3r}2\right\rfloor, \dots,  \frac{n}2\right\}$$ the collection of all medians of these intervals (except for the last one). The set $\set{Med}$ shall be used as an index set for the intervals. 
 
 Now, we define the compressor for unknown spectrum as in the following:
\begin{itemize}
\item {\em Encoder.} 
First perform the Schur transform and measures the state of the index register. For outcome $J$, store the index $i(J)$ in a classical memory so that $J\in\set{L}_{i(J)}$.
Apply the operation $\map{C}_{J\to f(J)}$  to the state of the representation register correspondingly, where 
$f$ is a function mapping any $J\in\{0,\dots,n/2\}$ to the median of the subset containing $J$:
\begin{align*}
f: J\to J_{\rm med}\in\set{Med}\quad{\rm s.t.}\quad J_{\rm med}\in\set{L}_{i(J)}.
\end{align*}
Then apply the frequency projection $\map{P}_{{\rm proj}, f(J)}$ (\ref{define-setC}). Encode the state after projection into a quantum memory.

\item {\em Decoder.} Read the classical memory for the value of $i(J)$ and sample a value $K$ uniformly in the subset $\set{L}_{i(J)}$. Apply the operation $\map{C}_{f(J)\to K}$ to the state of the quantum memory. Then append the state $I_{m_{K}}/m_{K}$ to the output. Finally, perform the inverse of the Schur transform to get the recovered state. 

\end{itemize}
The compressor requires a quantum memory of $\log (\sqrt{J}\log J+1)$ qubits, which is upper bounded by $(1/2)\log n+\log\log n+1$. It also requires a classical memory of $\log b\le (1/2+x)\log n$ bits.

\medskip
\noindent{\bf Error analysis.}
To avoid redundancy, we only analyze the error in the (more complex) case where the spectrum is unknown, while the error for fixed spectrum can be analyzed in the same way with less intermediate steps. The error bound for unknown spectrum also holds for fixed spectrum.

\begin{Lemma}
The error of compression can be bounded as 
\begin{align}
\epsilon_{t,p}\le \left(\frac{n}{2p-1}\right)^{-1/2}+\frac32 [(p-1/2)n]^{-\frac18\ln\frac{p}{1-p}}.
\end{align} 
\end{Lemma}
\begin{proof}
The output state of the unknown spectrum compressor $(\map{E},\map{D})$ can be derived by inserting the above decomposition into the protocol description. The state is expressed as
\begin{align}\label{output}
&\map{D}\circ\map{E}\left(\rho_{t,p}^{\otimes n}\right)=\\
&\sum_{J}\sum_{K\in\set{L}_{i(J)}}\frac{q_J|K\>\<K|}{|\set{L}_{i(J)}|} \otimes\map{C}_{f(J)\to K}\circ\map{P}_{{\rm proj},f(J)}\circ\map{C}_{J\to f(J)}\left(\rho_{t,p,J}\right)\otimes\frac{I_{m_K}}{m_K}. \nonumber
\end{align}
To bound the error $\epsilon_{t,p}=\frac12\left\|\map{D}\circ\map{E}\left(\rho_{t,p}^{\otimes n}\right)-\rho_{t,p}^{\otimes n}\right\|_1$, we first use the concentration property of $q_J$. Notice that the probability distribution $q_J$ in Eq. (\ref{decomp}) has the explicit form \cite{universal}
\begin{align}
q_{J}=\frac{2J+1}{2J_0}&\left[  B\left(\frac n2 +  J+1 \right) -B\left(\frac n2 - J \right)\right]\label{qJ}
\end{align}
where $B(k)=p^{k}(1-p)^{n-k}{n\choose k}$ and $J_0 = (p-1/2)(n+1)$, which is a Gaussian distribution concentrated in an interval of width $O(\sqrt{n})$ around $J_0$ when $n$ is large. We define the following interval
\begin{align*}
\set{C}=\left\{\lfloor J_0-n^{(1+x)/2}\rfloor,\dots,\lfloor J_0+n^{(1+x)/2}\rfloor\right\}.
\end{align*}
For $J$ outside of $\set{C}$, we use the upper bound of trace distance $\|\cdot\|_1\le2$ to bound the error as $\sum_{J\not\in\set{C}}q_J$, while for $J$ in $\set{C}$ the error can be further split into the error term of the interpolation (namely the difference between the uniformly sampled distribution and $q_J$), the error term of the frequency projection, and the error term of $\map{C}_{J\to K}$.
The interpolation term can be bounded using the smoothness of $q_J$; the term for the frequency projection can be bounded using Lemma \ref{lemma-proj}; and the term for $\map{C}_{J\to K}$ can be bounding using the following result  (see Lemma 1 of \cite{universal}):
 $\map{C}_{J\to K}$   
 transforms $\rho_{t,p,J}$ into $\rho_{t,p,K}$ with  error
 \begin{align}\label{e-bound}
  \left\|      \map C_{J\to K}  \left(\rho_{t,p,J}\right)  -   \rho_{t,p,K} \right\|_1 \le \delta^{1-x}+O\left(\delta \right) \, ,  
\end{align}
where $x>0$ is an arbitrary constant and $\delta:={{|J-K|}/{J}}$.

Now we analyze the error following the aforementioned idea. We first express it as
\begin{align}
\epsilon_{t,p}&=\frac12\left\|\map{D}\circ\map{E}\left(\rho_{t,p}^{\otimes n}\right)-\rho_{t,p}^{\otimes n}\right\|_1\nonumber\\
&\le \frac12\left\|\map{D}\circ\map{E}\left(\rho_{t,p}^{\otimes n}\right)-\rho_{t,n}\right\|_1+\frac12\left\|\rho_{t,n}-\rho_{t,p}^{\otimes n}\right\|_1\label{error-inter1}
\end{align}
where $\rho_{t,n}$ is the intermediate state
\begin{align}\label{intermediate}
\rho_{t,n}:=\sum_{J\in\set{C}}\left(\sum_{K\in\set{L}_{i(J)}}\frac{q_{K}}{|\set{L}_{i(J)}|}\right)|J\>\<J|\otimes\rho_{t,p,J}\otimes\frac{I_{m_J}}{m_J}
\end{align}
defined for convenience. We rewrite Eq. (\ref{output}) as 
\begin{align*}
\map{D}\circ\map{E}\left(\rho_{t,p}^{\otimes n}\right)=&\sum_{J}\left(\frac{\sum_{K\in\set{L}_{i(J)}}q_K}{|\set{L}_{i(J)}|}\right) |J\>\<J|\\
&\otimes\map{C}_{f(J)\to J}\circ\map{P}_{{\rm proj},f(J)}\circ\map{C}_{K\to f(J)}\left(\rho_{t,p,K}\right)\otimes\frac{I_{m_J}}{m_J} 
\end{align*}
using the fact that the ranges of $J$ and $K$ are the same and $f(K)=f(J)$ for $K\in\set{L}_{i(J)}$.
Then the first term in Eq. (\ref{error-inter1}) can be bounded as
\begin{align*}
&\frac12\left\|\map{D}\circ\map{E}\left(\rho_{t,p}^{\otimes n}\right)-\rho_{t,n}\right\|_1\nonumber\\
\le&\max_{J\in\set{C}}\max_{K\in\set{L}_{i(J)}}\frac12\left\|\map{C}_{f(J)\to J}\circ\map{P}_{{\rm proj},f(J)}\circ\map{C}_{K\to f(J)}\left(\rho_{t,p,K}\right)-\rho_{t,p,J}\right\|_1\nonumber\\
&+\sum_{J\not\in\set{C}}\sum_{K\in\set{L}_{i(J)}}\frac{q_K}{|\set{L}_{i(J)}|}\\
\le&\max_{J\in\set{C}}\max_{K\in\set{L}_{i(J)}}\left\{\left\|\map{C}_{K\to J}\left(\rho_{t,p,K}\right)-\rho_{t,p,J}\right\|_1\right.\\
&\left.+\frac12\left\|\map{P}_{{\rm proj},J}\left(\rho_{t,p,J}\right)-\rho_{t,p,J}\right\|_1\right\}+\sum_{J\not\in\set{C}}\sum_{K\in\set{L}_{i(J)}}\frac{q_K}{|\set{L}_{i(J)}|}
\end{align*}
having used the monotonicity, the upper bound and the triangle inequality of trace distance. Applying Eq. (\ref{e-bound}) and Lemma \ref{lemma-proj}, we have
\begin{align*}
&\frac12\left\|\map{D}\circ\map{E}\left(\rho_{t,p}^{\otimes n}\right)-\rho_{t,n}\right\|_1\\
\le& \left(\frac{r}{J_0}\right)^{1-x'}+\frac32 (J_0)^{-\frac18\ln\frac{p}{1-p}}+\sum_{|J-J_0|>n^{1/2+x}}\sum_{K\in\set{L}_{i(J)}}\frac{q_K}{|\set{L}_{i(J)}|}
\end{align*}
for every $x'>0$, having used $J_0=O(n)$. The last term can be bounded using the tail property of $q_J$, and we have
\begin{align}
&\frac12\left\|\map{D}\circ\map{E}\left(\rho_{t,p}^{\otimes n}\right)-\rho_{t,n}\right\|_1\nonumber\\
\le& \left(\frac{r}{J_0}\right)^{1-x'}+\frac32 (J_0)^{-\frac18\ln\frac{p}{1-p}}+\sum_{|J-J_0|>n^{1/2+x}-r}q_J\nonumber\\
\le&\left(\frac{r}{J_0}\right)^{1-x'}+\frac32 (J_0)^{-\frac18\ln\frac{p}{1-p}}+2\exp\left(-\frac{2n^{x}}{p^2}\right),\label{error-inter2}
\end{align}
where the last step comes from Hoeffding's inequality.
Next, substituting Eqs. (\ref{decomp}) and (\ref{intermediate}) into the second error term in Eq. (\ref{error-inter1}), we can bound it as 
\begin{align}
\frac12\left\|\rho_{t,n}-\rho_{t,p}^{\otimes n}\right\|_1&\le \frac12\sum_{J\in\set{C}}\left|q_J-\left(\sum_{K\in\set{L}_{i(J)}}\frac{q_{K}}{|\set{L}_{i(J)}|}\right)\right|+\sum_{J\not\in\set{C}}q_J\nonumber\\
&\le\frac12\sum_{J\in\set{C}}\sum_{K\in\set{L}_{i(J)}}\frac{q_K}{|\set{L}_{i(J)}|}\left|\frac{q_J}{q_K}-1\right|+\sum_{J\not\in\set{C}}q_J\nonumber\\
&\le\frac12\max_{J\in\set{C}}\max_{K\in\set{L}_{i(J)}}\left|\frac{q_J}{q_K}-1\right|+\sum_{J\not\in\set{C}}q_J.\label{error-inter21}
\end{align}
Now, by Eq. (\ref{qJ}) we have
\begin{align*}
\frac{q_K}{q_{J}}&=\frac{2K+1}{2J+1}\cdot\frac{B\left(\frac{n}2+K+1\right)-B\left(\frac{n}2-K\right)}{B\left(\frac{n}2+J+1\right)-B\left(\frac{n}2-J\right)}.
\end{align*}
We further notice that, by the De Moivre-Laplace theorem, the binomial $B(k)$ can be approximated by a Gaussian for $J\in\set{C}$ and for large $n$. Precisely we have
$$B\left(\frac{n}2+J+1\right)=\frac{\exp\left[-\frac{(J-J_0)^2}{2np(1-p)}\right]}{\sqrt{2\pi np(1-p)}}\left[1+O\left(\frac{1}{\sqrt{n}}\right)\right].$$
Moreover, noticing that the term $B\left(\frac{n}2-J\right)$ is exponentially small compared to $B\left(\frac{n}2+J+1\right)$, we have
\begin{align}\label{step-e2-21}
\frac{q_K}{q_{J}}&\le\frac{J_0-n^{(1+x)/2}}{J_0-n^{(1+x)/2}-r}\left[1-O\left(n^{-x/2}\right)\right]\\
  \frac{q_K}{q_{J}}&\ge\frac{J_0-n^{(1+x)/2}-r}{J_0-n^{(1+x)/2}}\left[1-O\left(n^{-x/2}\right)\right]\nonumber
\end{align}
Applying Eq. (\ref{step-e2-21}) to the first term and Hoeffding's inequality to the second term in Eq. (\ref{error-inter21}), we have
\begin{align}\label{error-inter3}
\frac12\left\|\rho_{t,n}-\rho_{t,p}^{\otimes n}\right\|_1&\le\frac{r}{2J_0}+2\exp\left(-\frac{2n^x}{p^2}\right).
\end{align}
Finally, substituting Eqs. (\ref{error-inter2}) and (\ref{error-inter3}) into Eq. (\ref{error-inter1}) and noticing $r=n^{1/2-x}/2$ and $J_0=(p-1/2)(n+1)$, we get the upper bound of the overall error as
\begin{align}
\epsilon_{t,p}\le \left(\frac{n}{2p-1}\right)^{-1/2}+\frac32 [(p-1/2)n]^{-\frac18\ln\frac{p}{1-p}}.
\end{align} Note that we kept only leading-order terms.

\end{proof}

\section{Strong converse of clock compression.}
Here we show the optimality of our clock compressors, by proving a strong converse for clock compression which extends the result for classical probability distributions \cite{classical} to the quantum regime. The strong converse states that, if a compressor $(\map{E},\map{D})$ requires a quantum memory of dimension only $d_{\rm enc}=O(n^{1/2-\delta})$ ($\delta>0$), then the expected error $\mathbb{E}_t\left[\frac12\left\|\map{D}\circ\map{E}\left(\rho_{t,p}^{\otimes n}\right)-\rho_{t,p}^{\otimes n}\right\|_1\right]$ (relative to the uniform distribution over $\varphi$) converges to 1 the  large $n$ limit.   The proof proceeds   by contradiction.  Let us assume that the error does not converge to 1, namely
\begin{align}
\mathbb{E}_t\left[\frac12\left\|\map{D}\circ\map{E}\left(\rho_{t,p}^{\otimes n}\right)-\rho_{t,p}^{\otimes n}\right\|_1\right]\le 1-\Delta
\end{align}
for some $\Delta>0$ and for every $n$. Applying Markov's inequality, we immediately get that
\begin{align*}
&\mathbf{Prob}\left[\frac12\left\|\map{D}\circ\map{E}\left(\rho_{t,p}^{\otimes n}\right)-\rho_{t,p}^{\otimes n}\right\|_1\le 1-\Delta/4\right]\ge\frac{\Delta}{4-\Delta},
\end{align*}
which implies that there exists $5d_{\rm enc}/\Delta$ points $\{t_i\}$ in $[0,2\pi)$ satisfying 
$$|t_i-t_j|\ge \frac{2\pi\Delta}{4-\Delta}\left(\frac{5d_{\rm enc}}{\Delta}\right)^{-1}$$
for $i\not=j$ and 
$$\frac12\left\|\map{D}\circ\map{E}\left(\rho_{t_i,p}^{\otimes n}\right)-\rho_{t_i,p}^{\otimes n}\right\|_1\le 1-\Delta/4$$
for any $i$. Moreover, states in the set $\{\rho_{t_i,p}^{\otimes n}\}$ can be distinguished with arbitrarily high precision: since $|t_i-t_j|=\Omega(n^{-1/2+\delta})$, there exists a POVM $\{O_i\}_{i=1}^{5d_{\rm enc}/\Delta}$ satisfying that, for arbitrarily small $\epsilon>0$ and for any $i$,
$$\Tr\left[\rho_{t_i,p}^{\otimes n}\,O_i\right]> 1-\epsilon$$
for large enough $n$ (see, for instance, Theorem 1 of \cite{tomography}). Now, by definition of the trace norm we have
\begin{align*}
\Tr\left[\map{D}\circ\map{E}\left(\rho_{t_i,p}^{\otimes n}\right)O_i\right]&>1-\epsilon-\frac12\left\|\map{D}\circ\map{E}\left(\rho_{t_i,p}^{\otimes n}\right)-\rho_{t_i,p}^{\otimes n}\right\|_1\\
&\ge \Delta/4-\epsilon\qquad\forall\, i.
\end{align*}
Using this property we have
\begin{align*}
d_{\rm enc}&=\Tr[\map{D}(I_{\rm enc})]\\
&=\sum_{i=1}^{5d_{\rm enc}/\Delta}\Tr[\map{D}(I_{\rm enc})O_i]\\
&\ge\sum_{i=1}^{5d_{\rm enc}/\Delta}\Tr\left[\map{D}\circ\map{E}\left(\rho_{t_i,p}^{\otimes n}\right)O_i\right]\\
&> \frac{5d_{\rm enc}}{\Delta}\left(\frac{\Delta}4-\epsilon\right),
\end{align*}
where $I_{\rm enc}$ denotes the identity on the encoding subspace. Remember that $\epsilon$ can be made arbitrarily small, and the above inequality leads to a contradiction when we set, for example, $\epsilon=\Delta/20$.

\section{Conclusion.}
In this work, we studied the compression of qubit clock states. We present optimal compressors for $n$ identical clock qubits that require only $(1/2)\log n$ qubits, which is half of the quantum memory cost of the general qubit compression \cite{universal}. The classical memory cost has also been cut down compared to the protocol in \cite{stopwatch}. Our results can be  applied to compress probes in quantum parameter estimation or to build quantum sensors. It remains an open question how these results can be generalized to clocks of higher or infinite dimensions.

\subsection*{Acknowledgments}
This work is supported by the Canadian Institute for Advanced Research
(CIFAR), by the Hong Kong Research Grant Council
through Grant No. 17300317, by National Science
Foundation of China through Grant No. 11675136,
by the HKU Seed Funding for Basic Research, and by the Foundational Questions Institute through grant FQXi-RFP3-1325.
MH was supported in part by JSPS Grants-in-Aid for Scientific Research 
(A) No.17H01280 and (B) No. 16KT0017.
\bibliographystyle{unsrt}
\bibliography{ref}

\begin{thebibliography}{10}

\bibitem{clock-standard}
L~Essen and JVL Parry.
\newblock An atomic standard of frequency and time interval: a caesium
  resonator.
\newblock {\em Nature}, 176:280, 1955.

\bibitem{optical-clock}
Andrew~D Ludlow, Martin~M Boyd, Jun Ye, Ekkehard Peik, and Piet~O Schmidt.
\newblock Optical atomic clocks.
\newblock {\em Reviews of Modern Physics}, 87(2):637, 2015.

\bibitem{chuang-clock}
Isaac~L Chuang.
\newblock Quantum algorithm for distributed clock synchronization.
\newblock {\em Physical Review Letters}, 85(9):2006, 2000.

\bibitem{jozsa-clock}
Richard Jozsa, Daniel~S Abrams, Jonathan~P Dowling, and Colin~P Williams.
\newblock Quantum clock synchronization based on shared prior entanglement.
\newblock {\em Physical Review Letters}, 85(9):2010, 2000.

\bibitem{lukin-clock}
Eric~M Kessler, Peter Komar, Michael Bishof, Liang Jiang, Anders~S S{\o}rensen,
  Jun Ye, and Mikhail~D Lukin.
\newblock Heisenberg-limited atom clocks based on entangled qubits.
\newblock {\em Physical Review Letters}, 112(19):190403, 2014.

\bibitem{network-nature}
Peter Komar, Eric~M Kessler, Michael Bishof, Liang Jiang, Anders~S S{\o}rensen,
  Jun Ye, and Mikhail~D Lukin.
\newblock A quantum network of clocks.
\newblock {\em Nature Physics}, 10(8):582--587, 2014.

\bibitem{network-prl}
P.~K\'om\'ar, T.~Topcu, E.~M. Kessler, A.~Derevianko,
  V.~Vuleti\ifmmode~\acute{c}\else \'{c}\fi{}, J.~Ye, and M.~D. Lukin.
\newblock Quantum network of atom clocks: A possible implementation with
  neutral atoms.
\newblock {\em Physical Review Letters}, 117:060506, Aug 2016.

\bibitem{gps}
Vittorio Giovannetti, Seth Lloyd, and Lorenzo Maccone.
\newblock Quantum-enhanced positioning and clock synchronization.
\newblock {\em Nature}, 412(6845):417--419, 2001.

\bibitem{fs}
Giorgio Santarelli, Ph~Laurent, Pierre Lemonde, Andr{\'e} Clairon, Anthony~G
  Mann, S~Chang, Andre~N Luiten, and Christophe Salomon.
\newblock Quantum projection noise in an atomic fountain: A high stability
  cesium frequency standard.
\newblock {\em Physical Review Letters}, 82(23):4619, 1999.

\bibitem{astronomy1}
Tilo Steinmetz, Tobias Wilken, Constanza Araujo-Hauck, Ronald Holzwarth,
  Theodor~W H{\"a}nsch, Luca Pasquini, Antonio Manescau, Sandro D'Odorico,
  Michael~T Murphy, Thomas Kentischer, et~al.
\newblock Laser frequency combs for astronomical observations.
\newblock {\em Science}, 321(5894):1335--1337, 2008.

\bibitem{astronomy2}
Chih-Hao Li, Andrew~J Benedick, Peter Fendel, Alexander~G Glenday, Franz~X
  K{\"a}rtner, David~F Phillips, Dimitar Sasselov, Andrew Szentgyorgyi, and
  Ronald~L Walsworth.
\newblock A laser frequency comb that enables radial velocity measurements with
  a precision of 1 cm $s^{-1}$.
\newblock {\em Nature}, 452(7187):610--612, 2008.

\bibitem{yang-chiribella-2016-prl}
Yuxiang Yang, Giulio Chiribella, and Daniel Ebler.
\newblock Efficient quantum compression for ensembles of identically prepared
  mixed states.
\newblock {\em Physical Review Letters}, 116:080501, Feb 2016.

\bibitem{universal}
Yuxiang Yang, Giulio Chiribella, and Masahito Hayashi.
\newblock Optimal compression for identically prepared qubit states.
\newblock {\em Physical Review Letters}, 117:090502, Aug 2016.

\bibitem{stopwatch}
Yuxiang Yang, Giulio Chiribella, and Masahito Hayashi.
\newblock Quantum stopwatch: How to store time in a quantum memory.
\newblock {\em arXiv preprint arXiv:1703.05876}, 2017.

\bibitem{population}
Yuxiang Yang, Ge~Bai, Giulio Chiribella, and Masahito Hayashi.
\newblock Compression for quantum population coding.
\newblock {\em IEEE Transactions on Information Theory}, PP(99):1--1, 2018.

\bibitem{bacon-chuang-2006-prl}
Dave Bacon, Isaac~L Chuang, and Aram~W Harrow.
\newblock Efficient quantum circuits for {S}chur and {C}lebsch-{G}ordan
  transforms.
\newblock {\em Physical Review Letters}, 97(17):170502, 2006.

\bibitem{harrow2005applications}
Aram~W Harrow.
\newblock Applications of coherent classical communication and the schur
  transform to quantum information theory.
\newblock {\em arXiv preprint quant-ph/0512255}, 2005.

\bibitem{fulton-harris}
William Fulton and Joe Harris.
\newblock {\em Representation theory}, volume 129.
\newblock Springer Science \& Business Media, 1991.

\bibitem{book-hayashi}
Masahito Hayashi.
\newblock {\em Group Representation for Quantum Theory}.
\newblock Springer, 2017.

\bibitem{werner}
Reinhard~F Werner.
\newblock Optimal cloning of pure states.
\newblock {\em Physical Review A}, 58(3):1827, 1998.

\bibitem{classical}
Masahito Hayashi and Vincent Tan.
\newblock Minimum rates of approximate sufficient statistics.
\newblock {\em IEEE Transactions on Information Theory}, 64(2):875--888, Feb
  2018.

\bibitem{tomography}
Jeongwan Haah, Aram~W. Harrow, Zhengfeng Ji, Xiaodi Wu, and Nengkun Yu.
\newblock Sample-optimal tomography of quantum states.
\newblock In {\em Proceedings of the Forty-eighth Annual ACM Symposium on
  Theory of Computing}, STOC '16, pages 913--925, New York, NY, USA, 2016. ACM.

\end{thebibliography}

\end{document}